\pdfoutput=1

\documentclass[journal]{IEEEtran}
\ifCLASSINFOpdf
  \usepackage[pdftex]{graphicx}
  \graphicspath{{./}{./figures/}}
  \DeclareGraphicsExtensions{.pdf,.jpeg,.png}
\else
\fi
\newcommand{\norm}[1]{\Vert#1\Vert}

\usepackage{amsfonts}
\usepackage{amsthm}
\usepackage{amsmath}
\usepackage{enumitem}

\usepackage{cite}

\newtheorem{thm}{Theorem}

\newcommand{\hl}[1]{#1}

\begin{document}
%
\title{Self-supervised Dynamic CT Perfusion Image Denoising with Deep Neural Networks}
%
%
%

\author{Dufan Wu,
        Hui Ren,
        and Quanzheng Li
\thanks{This work was supported in part by NIH under grant 1RF1AG052653 and 5P41EB022544.

D. Wu, H. Ren, and Q. Li are with the Center for Advanced Medical Computing and Analysis and Gordon Center for Medical Imaging, Massachusetts General Hospital and Harvard Medical School, Boston, MA, 02114, USA. 

Emails: dwu6@mgh.harvard.edu; hren2@mgh.harvard.edu; li.quanzheng@mgh.harvard.edu}
}

%
%

\markboth{~Vol.~xx, No.~xx, xxxx}%
{Noise2Noise CTP Denoising}
%



\maketitle

\begin{abstract}
Dynamic computed tomography perfusion (CTP) imaging is a promising approach for acute ischemic stroke diagnosis and evaluation. Hemodynamic parametric maps of cerebral parenchyma are calculated from repeated CT scans of the first pass of iodinated contrast through the brain. It is necessary to reduce the dose of CTP for routine applications due to the high radiation exposure from the repeated scans, where image denoising is necessary to achieve \hl{a} reliable diagnosis. \hl{In this paper,} we proposed a self-supervised deep learning method for CTP denoising, which did not require any high-dose reference images for training. The network was trained by mapping each frame of CTP to an estimation from its adjacent frames. \hl{Because the noise} in the source and target was independent, this approach could effectively remove the noise. Being free from high-dose training images granted the proposed method easier adaptation to different scanning protocols. The method was validated on both simulation and a public real dataset. The proposed method achieved improved image quality compared to conventional denoising \hl{methods}. On the real data, the proposed method also had improved spatial resolution and contrast-to-noise ratio compared to supervised learning which was trained on the simulation data. \bstctlcite{IEEEexample:BSTcontrol}
\end{abstract}


%
\IEEEpeerreviewmaketitle

\section{Introduction}
%
%
%
%
\IEEEPARstart{S}{troke} is the 5th cause of death and a leading cause of long term disability in the United States \cite{Benjamin2017}. Stroke is caused by the \hl{interruption} of blood supply to part of the cerebral tissue, which leads to \hl{a lack} of oxygen in the tissue and permanent brain damage. Ischemic stroke which is caused by \hl{an obstructed} blood supply, accounts for approximately 87\% of all the strokes \cite{Benjamin2017}. Mechanical thrombectomy has been proved to be an effective treatment for certain patients suffering from \hl{an} ischemic stroke within 6 to 24 hours from symptom onset \cite{Nogueira2018}. Findings from imaging, such as the size of infarct cores, are important criteria to determine the patients' eligibility. Hence, imaging plays an important role in ischemic stroke management and dynamic computed tomography perfusion (CTP) is among the I-A recommendations (strong recommendation with high-quality evidence) for eligible patients in the 2018 American Heart Association (AHA) guideline \cite{Powers2018}. Compared to other I-A imaging approaches such as \hl{diffusion-weighted} imaging (DWI), CTP has its major advantage on speed, availability, and cost-effectiveness compared to magnetic resonance imaging (MRI) \cite{Muir2006}. Because it is crucial to treat patients with ischemic stroke within 24 hours, the high availability of dedicated CT scanners in emergency departments is its main advantage over MR.

Dynamic cerebral CT perfusion repeatedly scans the brain during the first pass of \hl{the} iodinated contrast agent through cerebral parenchyma. The scan usually lasts for $60-75$ seconds with $\leq 3$ seconds interval between adjacent frames \cite{Konstas2009}. Hemodynamic parametric maps are computed from the time frames, including cerebral blood flow (CBF), cerebral blood volume (CBV), mean transit time (MTT) and time to peak (TTP). CBF and CBV are usually used to determine the infarct core where the occlusion of blood flow is most severe. MTT is highly correlated to the penumbra areas where the tissues have reduced blood flow and may be damaged as time elapses \cite{Konstas2009a}. The high radiation dose due to the continuous exposure to X-ray is one of the biggest concerns of CTP in clinical applications, and \hl{a low-dose} protocol is necessary to reduce \hl{the} potential risk of radiation \cite{Murphy2014}. Increased noise due to low-dose scans leads to very noisy parametric maps. There is even \hl{a} considerable amount of noise in parametric maps from normal-dose CTP scans \cite{Gonzalez2012}. Hence, image denoising plays an important role to achieve valid images for diagnosis. 

Conventional denoising algorithms \hl{for low-dose CT} are designed to reduce image noise with \hl{statistical modeling and} edge preservation priors\hl{\cite{Zhang2017, Zhao2019, Hasan2018, Geraldo2016}}. \hl{In CTP imaging, the structural correlation between time frames offers more information which can be utilized for denoising.} Time-intensity profile similarity (TIPS) and its variants use spatially variant filters according to similarities in both spatial and time domain \cite{Mendrik2011, Pisana2016, Pisana2017}. Gaussian process modeling has been proposed to utilize the smoothing prior of \hl{the time-concentration curve of the contrast agent} \cite{Zhu2012}. \hl{Besides} time frame denoising, penalty functions, such as tensor total variation (TTV) and sparse coding dictionary, have also been proposed to use with iterative deconvolution  \cite{Fang2013, Fang2015, Niu2016, Zeng2016}. 

\hl{The priors for image denoising can also be incorporated into the image reconstruction for improved image quality when raw data is accessible. Standard iterative reconstruction (IR) methods can be directly applied to each time-frame independently\cite{Negi2012, Lin2013, Niesten2014, Tao2014}. Special IR algorithms have also been developed for CTP by exploiting the inter-frame correlation. Low-noise prior images can be constructed from previous (non-contrast) scans or averaging the time frames, which will greatly improve the quality of time-frame image reconstruction\cite{Chen2008, Nett2010, Ma2012}. One can also reconstruct the difference between each time frame to "reduce" the number of unknowns\cite{Pourmorteza2016, Seyyedi2018}. The low-rank penalty can also be applied to the time-frame reconstruction\cite{Li2019}. Besides time-frame reconstruction, it is also possible to model the time-concentration curves with basis functions and directly reconstruct the parametric maps\cite{Manhart2013}. However, IR algorithms require access to raw data and are out of the scope of this work, where we aim at CTP denoising from images only. }

In recent years, deep learning-based medical image denoising has achieved great success, where it \hl{copes with} the complex structure and noise properties \hl{by} deep neural networks learned from data. \hl{The networks are trained to map the low-dose CT images to normal-dose CT images, usually on a training set with paired or unpaired low- and normal-dose images\cite{Wolterink2017, Chen2017, Kang2017, Chen2017a, Yang2018}.}
\hl{Besides} general low-dose CT denoising, there are a few applications of deep learning to CTP. Xiao et al. proposed a spatial-temporal neural network to map low-dose time frames to high-dose ones \cite{Xiao2019}. Kadimesetty et al. proposed to use deep learning to denoise parametric maps as well as time frames \cite{Kadimesetty2018}. Both studies demonstrated improved image quality compared to conventional algorithms. 

\hl{Despite the} promising performance of current deep-learning methods for denoising, they need high-dose reference images for training, whose acquisition is a non-trivial task. The various protocols of CTP also made it harder to acquire reference images to cover all the manufacturers and protocols \cite{Konstas2009}. Furthermore, it has been shown that there is \hl{still a considerable amount} of noise in the parametric maps of high-dose CTP \cite{Gonzalez2012}. Hence, it is desirable to train the denoising network without high-dose reference images. A recent work, Noise2Noise, demonstrated that denoising networks could be effectively trained by mapping between two independent noise realizations instead of mapping to clean images \cite{Lehtinen2018}. Wu et al. also applied the Noise2Noise to medical imaging where the two noise realizations were achieved by projection data splitting \cite{Wu2019}. However, in many situations, we do not have access to either two noise realizations of the same patients or the projection data. As a consequence, none of the aforementioned methods could be directly applied. \hl{Beside Noise2Noise, two other unsupervised learning frameworks for image denoising are also noticeable, including Deep Image Prior\cite{Ulyanov2018, Gong2019} and CycleGAN\cite{Kang2019}, which will be further discussed in section \ref{sec:existing}.}

To generate the two noise realizations, we exploited \hl{the fact} that adjacent time frames in CTP are noise independent but highly correlated in structure. An estimation of the current frame was approximated by averaging its adjacent two frames with linear correction, which was considered as another noise realization of the current frame. The denoising network was trained by mapping the frame images to their corresponding estimations. Hence, the proposed approach was self-supervised and did not require any additional data. An additional bias compensation term was added to the loss function, which took the loss between the \hl{low-frequency} components of denoised images and source images. The additional term also acted as an effective regularization term to prevent overfitting. The proposed Noise2Noise approach was validated on both simulation data and real CTP data from the 2018 Ischemic Stroke LEsion Segmentation (ISLES) challenge \cite{Aichert2013, Kistler2013, Maier2017}. It demonstrated improved performance compared to conventional methods including Gaussian filtering, TIPS filtering \cite{Mendrik2011} and TTV regularized deconvolution \cite{Fang2015} on both datasets. Compared to supervised learning, Noise2Noise achieved similar image quality on the simulation dataset, but improved spatial resolution and contrast-to-noise ratio (CNR) on the real dataset where \hl{the supervised network} was trained on the simulation data as in \cite{Kadimesetty2018}. 

\hl{
\section{Related Works}\label{sec:existing}

\subsection{Noise2Noise}
The main framework of Noise2Noise\cite{Lehtinen2018} is that when training denoising networks, instead of using clean images as the training target, one can also use a zero-mean and independent noise realization of the same object. In this work, we did not aim to challenge or modify this framework. Instead, we focused on how to obtain two such noise realizations from CTP time-frame images. In the original Noise2Noise work\cite{Lehtinen2018}, it was assumed that two such noise realizations are given. In \cite{Wu2019}, the authors further extended the application to single-sampled raw data, where the two independent realizations were constructed via raw data splitting. However, neither of the sampling strategies could be directly applied to the CTP images, which are single-sampled without raw data. 

\subsection{Deep Image Prior}
Deep Image Prior\cite{Ulyanov2018, Gong2019} achieved unsupervised single-image denoising by fitting random noise or prior images to noisy images through deep neural networks. The network will converge to the structures before the noise, and early stopping is used to remove the noise. Compared to Noise2Noise, the advantage of Deep Image Prior is the higher flexibility where only one noise realization is needed. However, its drawback is the requirement for network training during inference and sensitivity to parameter selection. Each CTP scan has many time-frame images, and Deep Image Prior may need to train a different network for each time frame in each CTP scan with different hyperparameters. Hence, Deep Image Prior is not an ideal choice for the CTP denoising. 

\subsection{CycleGAN}
A recent work\cite{Kang2019} used the cyclic adversarial loss to denoise retrospectively gated cardiac CT angiography (CTA) images. The CycleGAN was built to match the distribution of the noisy images at systolic phases with the clean images at diastolic phases. However, cerebral CTP images have significant differences with cardiac CTA, which are likely to jeopardize the efficacy of CycleGAN.

First, in cardiac CTA, the images at diastolic phases have lower noise than the ones in systolic phases because of the longer sampling time. However, such heavily uneven sampling is not available in cerebral CTP. 

Second, in cardiac CTA, the images between different time frames have similar structures and contrast levels, leaving the noise levels the biggest difference. On the other hand, in cerebral CTP, because of the dramatic change in iodine concentration, there are significant differences in the vessel visibility and image brightness among the early frames, frames near the peak of concentration and late frames. So even if we could construct a low-noise image by averaging time frames, it will have significant, intrinsic differences to the noisy time-frame images. These differences violate the requirement for CycleGAN-based denoising, that source and target images should have similar structures other than noise.

\section{Preliminaries}

\subsection{CTP Imaging}\label{sec:ctp_method}
Denote the time-frame images as $\mathbf{x}(1)$, $\mathbf{x}(2)$, $\dots$, $\mathbf{x}(t)$, $\dots$, $\mathbf{x}(T)$, the time-concentration images of iodinated contrast $\mathbf{c}(t)$ can be calculated as:
\begin{equation}\label{eq:tcc}
\mathbf{c}(t) = \mathbf{x}(t)- \frac{1}{T_0}\sum_{t=1}^{T_0}{\mathbf{x}(t)},
\end{equation}
where the second term is the estimation of non-contrast CT images with the average of early frames\cite{Fieselmann2011}. We used $T_0=2$ in our study. 

After $\mathbf{c}(t)$ is calculated for each voxel, parametric maps including CBF, CBV, MTT and TTP can be calculated via the deconvolution methods\cite{Fieselmann2011}. More details can be found in appendix \ref{app:ctp_method}.

\subsection{Noise2Noise Training}\label{sec:noise2noise}
Denote the denoising network as \hl{$f(\mathbf{x;\mathbf{\Theta}})$} which has input $\mathbf{x}$ and parameters to be learned $\mathbf{\Theta}$, Noise2Noise trains the denoising network with:
\begin{equation}\label{eq:n2n}
\mathbf{\Theta}^* = \arg\min_\mathbf{\Theta} \frac{1}{N}\sum_{i=1}^{N}\norm{f(\mathbf{x}_i+\mathbf{n}_{1i};\mathbf{\Theta}) - (\mathbf{x}_i+\mathbf{n}_{2i})}_2^2,
\end{equation}
which maps the $i$th noisy training image, $\mathbf{x}_i+\mathbf{n}_{1i}$ to another noise realization of it, $\mathbf{x}_i+\mathbf{n}_{2i}$. 

When $\mathbf{n}_{2i}$ is zero-mean and independent from $\mathbf{n}_{1i}$, the Noise2Noise denoising (\ref{eq:n2n}) is equivalent to training with clean images. Intuitively, it can be explained that because we cannot predict $\mathbf{n}_{2i}$ due to  its independence, the best way to minimize the L2-loss is predicting the mean of $\mathbf{n}_{2i}$, which is zero. A brief proof can be found in the appendix \ref{app:noise2noise}
}

\section{Methodology}
\subsection{Noise2Noise for CTP Denoising}
In CTP imaging, adjacent time frames are acquired in short time interval and similar \hl{to each other}, so we have:
\begin{equation}\label{eq:avg_relation}
\mathbf{c}(t) \approx \frac{\mathbf{c}(t-1) + \mathbf{c}(t+1)}{2}
\end{equation}

When building Noise2Noise loss from (\ref{eq:avg_relation}), it should be noted that although \hl{time-frame} images $\mathbf{x}(t-1)$, $\mathbf{x}(t)$ and $\mathbf{x}(t+1)$ \hl{have} independent and zero-mean noise, \hl{the averages of the early frame images in (\ref{eq:tcc}), $\sum_{t=1}^{T_0}\mathbf{x}(t)$, are the same and will introduce correlated noise.} 

To remove \hl{this noise dependence}, different early frames could be used on the \hl{l.h.s. and r.h.s.} of (\ref{eq:avg_relation}). Under the circumstance where \hl{$T_0=2$}, (\ref{eq:avg_relation}) became:
\begin{equation}\label{eq:avg_relation_2}
\mathbf{x}(t) - \mathbf{x}(1) \approx \frac{\mathbf{x}(t-1) + \mathbf{x}(t+1)}{2} - \mathbf{x}(2)
\end{equation}

Our Noise2Noise loss was built to map l.h.s. of (\ref{eq:avg_relation_2}) to its r.h.s. To compensate for the estimation's bias especially near the peak of \hl{the time-concentration curve}, we applied linear correction to \hl{the estimation}, which was:
\begin{equation}\label{eq:bc_train}
\mathbf{x}_{e}(t) = \kappa(t)\frac{\mathbf{x}(t-1) + \mathbf{x}(t+1)}{2}, 
\end{equation}
where 
\begin{equation}\label{eq:bc_train_2}
\kappa(t) = \arg\min_\kappa\norm{\kappa\frac{\mathbf{x}(t-1) + \mathbf{x}(t+1)}{2} - \mathbf{x}(t)}_2^2
\end{equation}

\hl{It is possible to use spatially variant $\kappa(t)$ for reduced bias, but it would introduce more hyperparameters such as the smoothness of $\kappa(t)$. Instead, we will introduce a more straight forward bias compensation term in section \ref{sec:bias_compensation} which is easier to tune. }

The Noise2Noise training loss was:
\begin{equation}\label{eq:n2n_loss}
\begin{split}
L_{n2n}&(\mathbf{\Theta}) = \frac{1}{2N}\sum_{i=1}^N\frac{1}{T_i-2}\sum_{t=2}^{T_i-1}\sum_{t_0=1}^2 \\
					&\norm{\hl{f(\mathbf{x}_i(t), \mathbf{x}_i(t_0);\mathbf{\Theta})} - \left( \mathbf{x}_{ie}(t) - \mathbf{x}_i(3-t_0) \right)}_2^2,
\end{split}
\end{equation}
where $T_i$ is the number of frames of sample $i$, $\mathbf{x}_{ie}(t)$ is the estimation of $\mathbf{x}_i(t)$ according to (\ref{eq:bc_train}) and (\ref{eq:bc_train_2}). \hl{$\mathbf{x}_i(t_0)$ and $\mathbf{x}_i(3-t_0)$ are the two different early frames,} which are $\mathbf{x}_i(1)$ and $\mathbf{x}_i(2)$ when $t_0=1$ and  $\mathbf{x}_i(2)$ and $\mathbf{x}_i(1)$ when $t_0=2$.

The network took one \hl{time-}frame image $\mathbf{x}(t)$ and one early frame image $\mathbf{x}(t_0)$ as input and directly output the denoised \hl{time-concentration image} $\mathbf{c}_d(t)$. During testing, $\mathbf{c}_d(t)$ was calculated by averaging the outputs with all possible $\mathbf{x}(t_0)$. In case of two early frames:
\begin{equation}
\mathbf{c}_d(t) = \frac{1}{2}\sum_{t_0=1}^2\hl{f(\mathbf{x}(t),\mathbf{x}(t_0);\mathbf{\Theta})}
\end{equation}

It should be mentioned that another possibility of building the Noise2Noise loss is doing it reversely, by mapping $\mathbf{c}(t-1)$ and $\mathbf{c}(t+1)$ to $\mathbf{c}(t)$. We chose the current approach mainly because of the following reasons:

First, \hl{the} network was easier to be applied to the time \hl{frames} since denoising only relied on the current frame. The network could be applied to the start and end of the time \hl{frames} without any data padding. 

Second, $\mathbf{x}(t)$ and $\mathbf{x}(t_0)$ contained all the structural information of $\mathbf{c}(t)$. \hl{Although $\mathbf{x}_e(t)$} might have relatively large bias near the peak of \hl{the} time-concentration curve, the bias was small at most time points where the time-concentration curve was monotonic. 

A diagram of the proposed Noise2Noise framework for CTP denoising is given in figure \ref{fig:scheme}. 

\begin{figure}[t]
\centering
\includegraphics[width=\columnwidth]{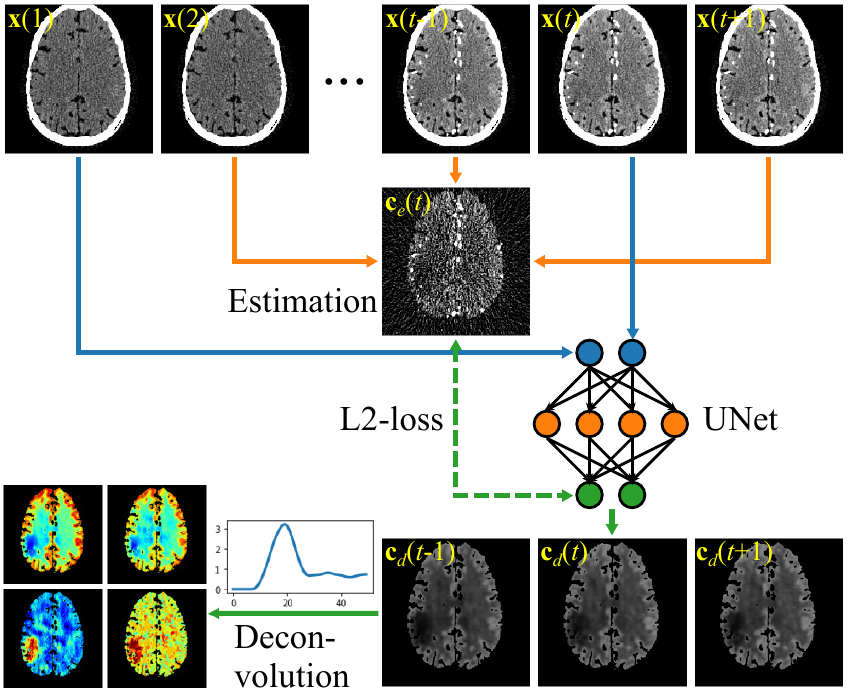}
\caption{Proposed Noise2Noise framework for CTP denoising. The blue lines demonstrate the inputs to the network. The orange lines demonstrate the estimation of $\mathbf{c}(t)$. Averaging over early frames are not demonstrated for the simplicity of the figure. }
\label{fig:scheme}
\end{figure}

\subsection{Bias Compensation}\label{sec:bias_compensation}
Bias could be introduced by \hl{several reasons:} bias of $\mathbf{x}_e(t)$; lack of training samples near the time-concentration peak compared to other time points; lack of total training samples which could lead to overfitting to noise. 

Since noise is mostly in the high frequency, an unbiased denoising algorithm should keep the low frequency of the noisy images unchanged. The bias compensation term was designed based on this assumption, which constrained the L2-distance between \hl{the} low-pass filtered time-concentration images and \hl{the} low-pass filtered output of network:
\begin{equation}\label{eq:bias_loss}
\hl{
\begin{aligned}
L&_{bias}(\mathbf{\Theta}) = \frac{1}{2N}\sum_{i=1}^N\frac{1}{T_i-2}\sum_{t=2}^{T_i-1}\sum_{t_0=1}^2 \\
					&\norm{\mathbf{G}*f(\mathbf{x}_i(t), \mathbf{x}_i(t_0);\mathbf{\Theta}) - \mathbf{G}*(\mathbf{x}_i(t) - \mathbf{x}_i(t_0))}_2^2,
\end{aligned}
}
\end{equation}
where $\mathbf{G}$ is a low-pass Gaussian filter. We used \hl{a very strong} low-pass filter (standard deviation of 6 pixels) to remove all the noise. 

The final training loss combined both Noise2Noise loss (\ref{eq:n2n_loss}) and bias compensation (\ref{eq:bias_loss}) and the network was trained as:
\begin{equation}\label{eq:total_loss}
\mathbf{\Theta}^* = \arg\min_\mathbf{\Theta}L_{n2n}(\mathbf{\Theta}) + \beta L_{bias}(\mathbf{\Theta}),
\end{equation}
where $\beta$ is a hyperparameter to balance between noise reduction and bias reduction. 

Lack of samples near \hl{the peaks of the time-concentration curve} could also lead to larger bias at these time points. However, the \hl{rapidly} changing images near \hl{the} peaks contain important information about hemodynamics and are potentially crucial to the accuracy of parametric maps. Hence, reducing the bias near \hl{the} peaks should have \hl{a} higher impact in reducing the bias of final parametric maps. To achieve this, we increased the sampling rate near \hl{the} peaks during training. In each batch, whereas half of the training data were sampled randomly along the time \hl{dimension}, the other half were sampled within a small window of width 5 near the peak. To determine the position of the peak, we summed all the pixels of interest and looked for the maximum position along time:
\begin{equation}
t_{peak} = \arg\max_t \mathbf{m}^T\mathbf{x}(t),
\end{equation}
where $\mathbf{m}$ is a thresholding mask which excluded bones and major vessels. \hl{These hyperparameters were selected by a few trials and errors to balance between bias compensation and convergence speed. Heavier sampling near the peak will lead to reduced bias near the peak but slower convergence. If the peak is too much oversampled, other parts of the time-concentration curves may be biased. However, the results are generally not sensitive to these hyperparameters. }

\section{Experimental Setup}
\subsection{Datasets}
\subsubsection{Simulation}\label{sec:simulation_data}
We used the \hl{open-source} code from \cite{Aichert2013} to generate simulation phantoms. Several ellipses were replaced inside the phantom to simulate infarct cores and penumbra areas of stroke. The phantom had \hl{an} axial resolution of $256 \times 256$ with \hl{a} pixel size of 1 mm. We used 50 continuous slices for training and another 15 continuous slices for testing. There was a 5-slice gap between the training and testing groups to reduce data correlation. 50 frames of CTP images were simulated with 1 second time interval. 

To generate CTP images under different noise \hl{levels}, we first forward projected the images into 2D sinograms with realistic geometry given in table \ref{table:geometry} \cite{Wu2017}. Poisson noise \hl{was} added to the sinogram according to:
\begin{equation}
p_{noisy} = -\log\left(\frac{\mathrm{Poisson}(N_0\exp\{-p\})}{N_0}\right),
\end{equation}
where $p$ is the forward projected value, and $N_0$ is the assumed number of initial photons for the each ray. The CTP images were then reconstructed from noisy sinograms via filtered backprojection (FBP) with Hann filter. Distance-driven projector \cite{Liu2017} was used for the forward projector and pixel-driven projector was used for the FBP. Three noise levels were simulated with $N_0=1\times10^5, 2\times10^5, 1\times10^6$. \hl{$1\times10^5$ and $2\times10^5$ were chosen to match the noise levels of our real CTP images. $1\times10^6$ was chosen for relatively low-noise images.} Some of the training images are given in figure \ref{fig:data}. 

\begin{table}[t]
\renewcommand{\arraystretch}{1.3}
\caption{Parameters of the Simulation Geometry}
\label{table:geometry}
\centering
\begin{tabular}{ll}
\hline
Parameter & Value \\
\hline
Geometry							& Equiangular fan-beam \\
Pixel size of image			& 1 $\times$ 1 mm$^2$  \\
Resolution of image			& 256 $\times$ 256 \\
Views per rotation 			& 1152 \\
Number of detector units 	& 384 \\
Pixel size of detector			& 1.2858 mm \\
Source to ios-center distance 		& 595 mm \\
Source to detector distance 	& 1086.5 mm \\
\hline
\end{tabular}
\end{table}

\begin{figure}[t]
\centering
\includegraphics[width=\columnwidth]{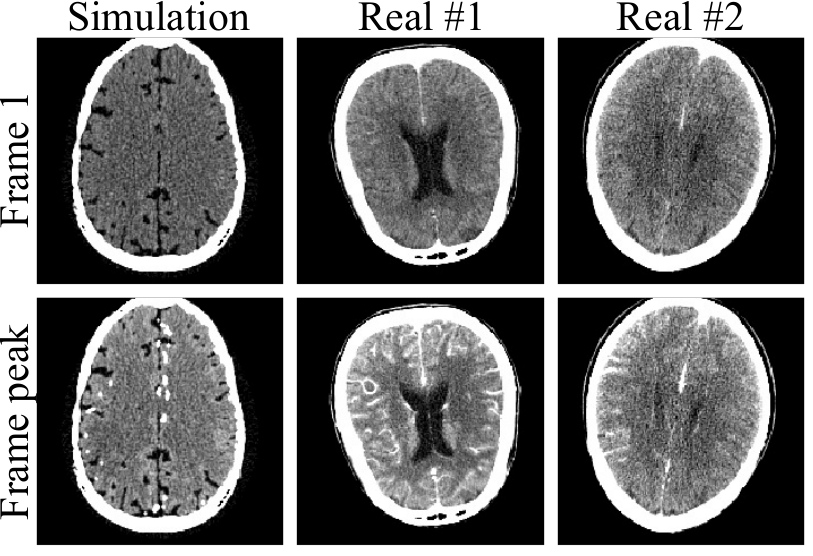}
\caption{Part of the training and testing dataset. The simulation images were under noise level of $N_0=2\times10^5$. The first row showed the \hl{first frame}, the second showed the frame at \hl{the peak of the time-concentration curve}. The display \hl{windows are} $40\pm80$ HU. }
\label{fig:data}
\end{figure}

\subsubsection{Real Data}
We used the CTP images from \hl{the 2018 ISLES} challenge dataset for real data validation \cite{Kistler2013, Maier2017}. 20 patients with similar noise appearance were used, where 16 \hl{were} used for training and 4 used for testing. Each patient had 2, 4, or 8 slices with 40 to 50 frames. All the patients were within 8 hours of symptom onset and a \hl{DWI was} done for each patient within 3 hours of CTP. Infarct cores derived from the DWI images were also provided. Some of the real data are illustrated in figure \ref{fig:data}.

\subsection{Preprocess and postprocess}
All the images were preprocessed by excluding bones and major blood vessels via thresholding. Bone masks were obtained by thresholding pixels larger than 120 HU in early frames. Only pixels inside the cortical bones were preserved. The major blood vessel masks were obtained by thresholding all the pixels with maximum intensity along time larger than 100 HU. 

We used \hl{the} ground truth \hl{aterial input function (AIF)} for parametric map calculation \hl{in} the simulation. For the real data, we used methods from \cite{Mouridsen2006, Kao2014} to automatically calculate venous output function (VOF) as well as AIF from unfiltered time frame images. Then AIF was scaled for partial volume correction by aligning its area under the curve with VOF. 

All the parametric maps were calculated \hl{by} methods in appendix \ref{app:ctp_method} after the time frames were denoised, except for TTV where denoising and deconvolution were done at the same time. 

\subsection{Quantitative metrics}\label{sec:metrics}
We used root-mean-square-error (RMSE) and structural similarity index (SSIM) \cite{Wang2004} of the parametric maps against noiseless images for the quantitative evaluation of simulation data. 

For the real data, we selected \hl{relatively} flat \hl{regions of interests (ROI)} inside normal white matter to calculate bias and standard deviation (std) of the denoised time-concentration images $\mathbf{c}_d(t)$. The bias was calculated against \hl{the} original noisy \hl{time-concentration images $\mathbf{c}(t)$}. Contrast-to-noise ratio (CNR) was also calculated for infarct cores against flat ROIs inside normal white matter for CBF images. The infarct cores were annotated on the CBF images referring to the regions derived from DWI. 

\subsection{Hyperparameters for Noise2Noise}
We used 2D UNet \cite{Ronneberger2015, Jin2017} as the backbone network for the Noise2Noise. The network was trained for 100 epochs with \hl{a} batch size of 20. The training algorithm was Adam \cite{Kingma2014} with \hl{a} learning rate of $10^{-4}$. For each batch, random time points from random slices were selected with oversampling near the peaks of \hl{the} time-concentration curves, according to section \ref{sec:bias_compensation}. $\beta$ was determined via parameter sweeping. The one that achieved \hl{the} best RMSE was selected for simulation, and the one that achieved \hl{the} best CNR was selected for real data. 

The image values were normalized to HU / 150 for \hl{the network input} and HU / 25 for \hl{the} output. Image augmentations were done by random \hl{flips} along x and y directions. 

\subsection{Comparison methods}
\subsubsection{No Filter}
Unfiltered parametric maps were directly calculated for each pixel via SVD with Tikhonov regularization \cite{Fieselmann2011}. 

\subsubsection{Gaussian Filter}
Spatially invariant Gaussian filter was applied to the time frames before deconvolution. The strength of \hl{the} filter was controlled with a single parameter $\sigma_g$.

\subsubsection{TIPS}
Time-intensity profile similarity (TIPS) \cite{Mendrik2011} is a bilateral filter where the weight between two pixels is determined by averaging their distances along the time. Two parameters corresponding to the filter strength along time and spatial domain, $\sigma_t$ and $\sigma_s$, were used in TIPS. 

\subsubsection{TTV}
Tensor total variation (TTV) \cite{Fang2015} applied TV prior to the pulse response functions $\mathbf{r}(t)$ (see appendix \ref{app:ctp_method}) in both time and spatial domain. To align the model bias with other methods, we added Tikhonov regularization to the loss function, which gave the following loss function for deconvolution:
\begin{equation}
\begin{split}
\mathbf{r}^* =  \arg\min_\mathbf{r} & \frac{1}{2}\norm{\mathbf{Ar - c}}_F^2 + \lambda\norm{\mathbf{r}}_F^2 \\
+ & \beta_s \norm{\nabla_x \mathbf{r}}_1 + \beta_s \norm{\nabla_y \mathbf{r}}_1 + \beta_t \norm{\nabla_t \mathbf{r}}_1,
\end{split}
\end{equation}
where $\mathbf{r}=(\mathbf{r}_1,\dots,\mathbf{r}_J)$ \hl{contains all the pulse response functions from the $J$ pixels} and $\mathbf{c}=(\mathbf{c}_1, \dots, \mathbf{c}_J)$ \hl{contains all the time-concentration curves}. $\mathbf{A}$ is built from AIF. $\nabla_x$, $\nabla_y$ and $\nabla_t$ are forward differential operators along the $x$, $y$ and $t$ directions. $\beta_s$ and $\beta_t$ are used to control the strength of TTV. 

\subsubsection{Supervised Learning}
Supervised learning was used to provide the reference to best possible performance using the same network structure in simulation. The same 2D UNet was trained to map noisy CTP frames to noiseless frames with loss function similar to (\ref{eq:n2n_loss}):
\begin{equation}
\begin{split}
L_{n2c}&(\mathbf{\Theta}) = \frac{1}{2N}\sum_{i=1}^N\frac{1}{T_i-2}\sum_{t=2}^{T_i-1}\sum_{t_0=1}^2 \\
					&\norm{\hl{f(\mathbf{x}_i(t), \mathbf{x}_i(t_0);\mathbf{\Theta})} - \mathbf{c}_{ig}(t)}_2^2,
\end{split}
\end{equation}
where
\begin{equation}
\hl{\mathbf{c}_{ig}(t) = \mathbf{x}_{ig}(t) - \frac{\mathbf{x}_{ig}(1) + \mathbf{x}_{ig}(2)}{2}}
\end{equation}
is the noiseless concentration \hl{map} at time $t$, \hl{and $\mathbf{x}_{ig}(1)$ and $\mathbf{x}_{ig}(2)$ are the two early frames to estimate non-contrast CT.} The bias compensation term (\ref{eq:bias_loss}) were also added to the supervised training loss, which could effectively prevent overfitting. 

Networks trained by supervised learning on $N_0=2\times10^5$ simulation data was further applied to denoise the real data. $N_0=2\times10^5$ \hl{was chosen} because the images had the most similar noise level to the real data. 

\subsubsection{Determination of Hyperparameters}
The hyperparameters of all the methods were determined by grid search. RMSE \hl{of CBF} was used as the selecting criteria for simulation data. For the real data, CNR \hl{of CBF} was used as the main criteria. However, it was found that the comparison methods would give strongly oversmoothed images at best CNR, so the final hyperparameter was adjusted by matching spatial resolution and noise level with the corresponding simulation results at $N_0=2\times10^5$. 

\section{Results}
\subsection{Simulation Results}

\begin{figure*}[t]
\centering
\includegraphics[width = 0.9\linewidth]{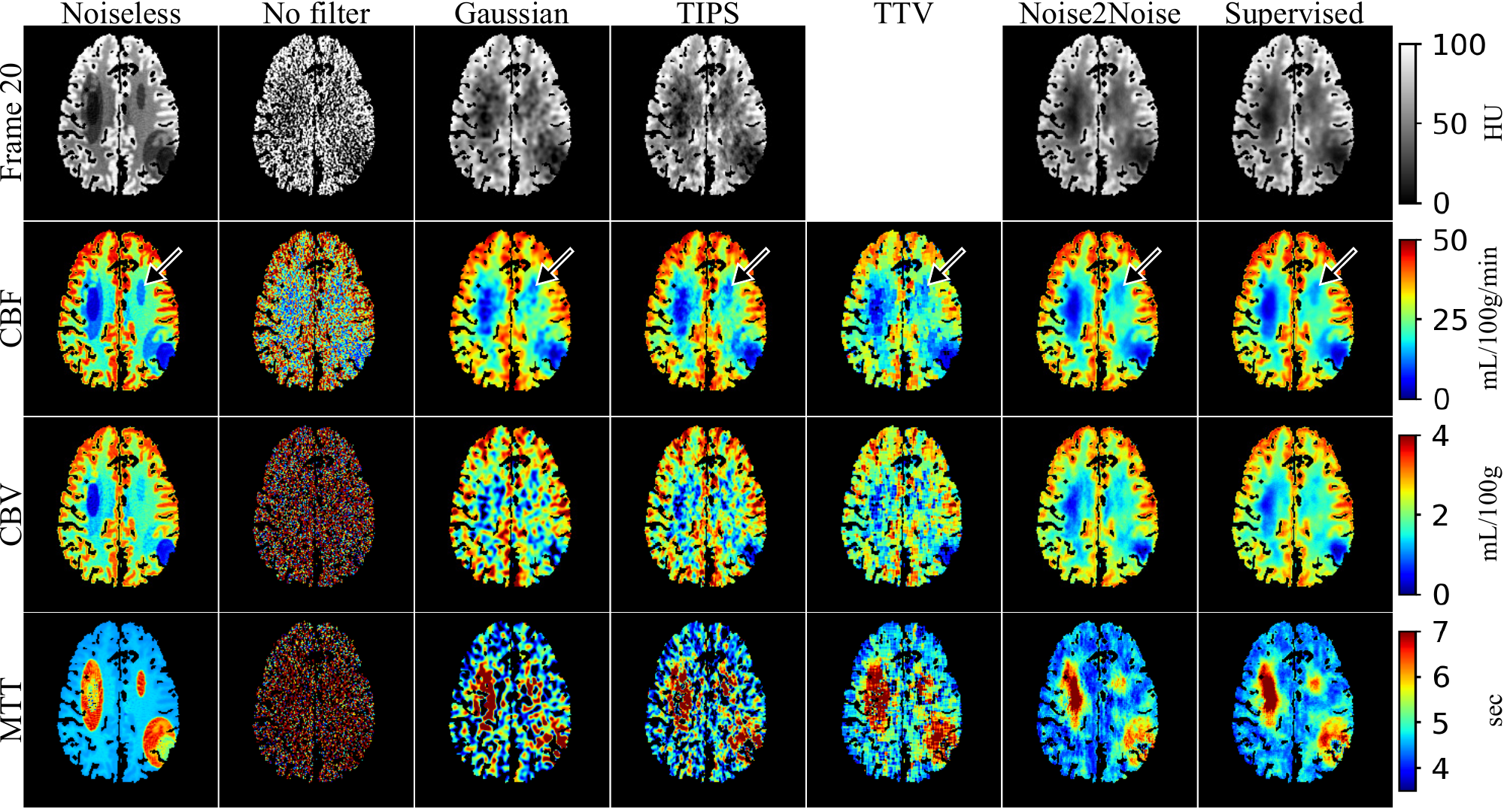}
\caption{Denoising results of all the algorithms for one of the testing simulation slices under $N_0=2\times10^5$. The concentration maps at the 20th frame (near the peak of AIF), CBF, CBV and MTT are given. The \hl{time-concentration} map is not given for TTV since TTV does not denoise the time frame images. }
\label{fig:simul_images}
\end{figure*}

\begin{figure}[t]
\centering
\includegraphics[width=\columnwidth]{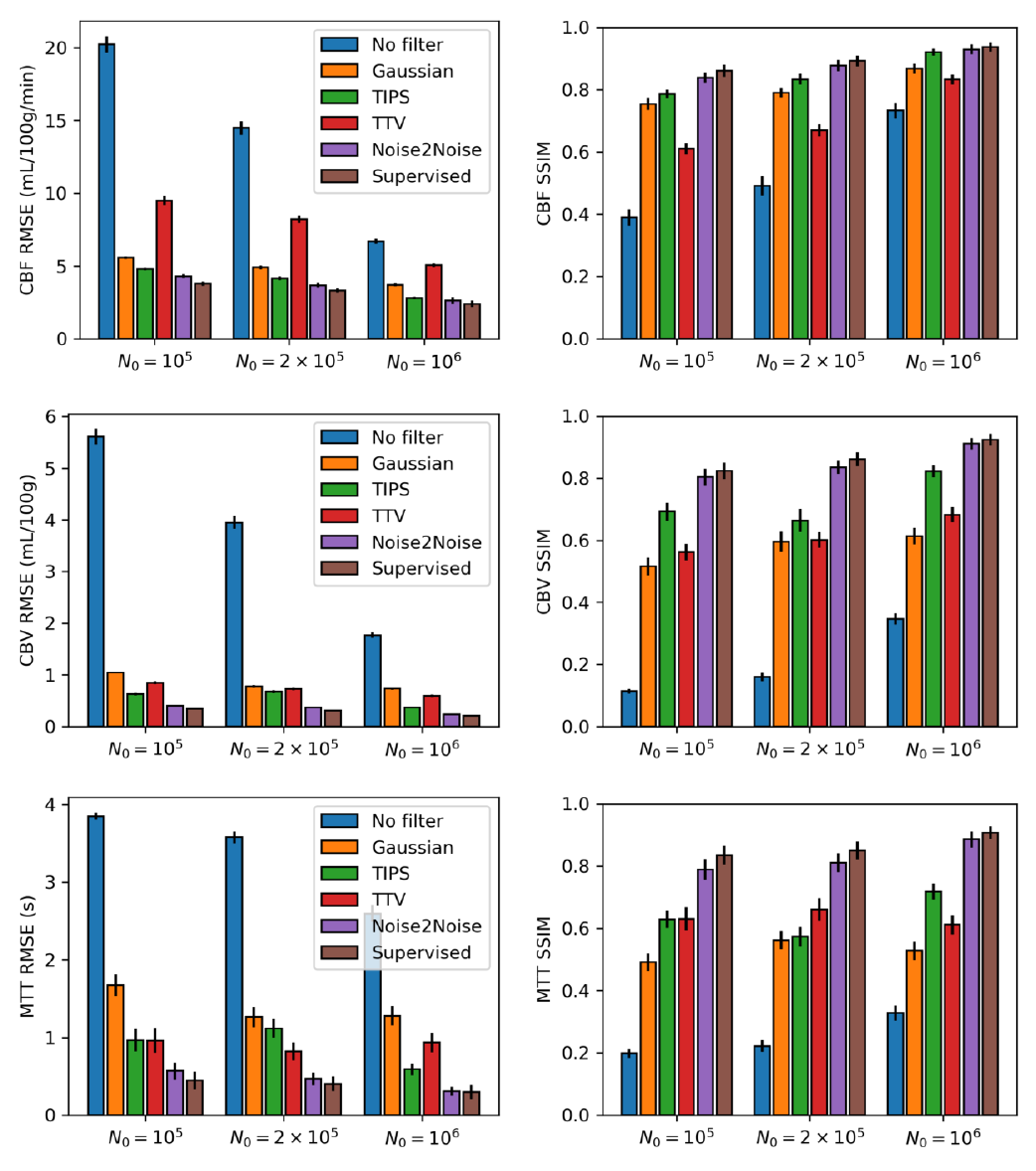}
\caption{RMSEs and SSIMs of CBF, CBV and MTT of testing simulation slices compared to noiseless results.  }
\label{fig:simul_quantification}
\end{figure}

\begin{figure}[t]
\centering
\includegraphics[width=\columnwidth]{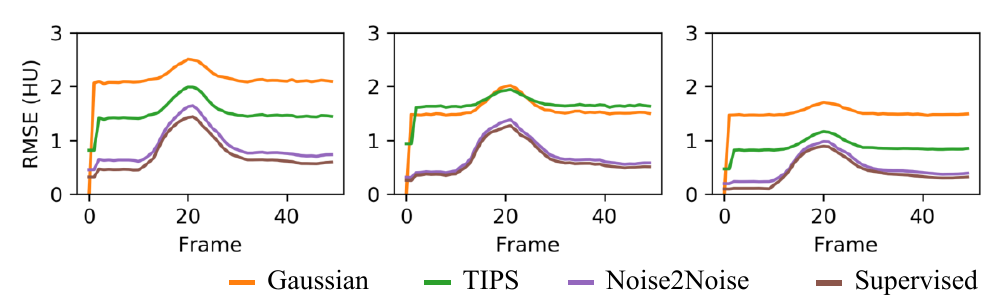}
\caption{RMSEs of frame images at each time point and noise level. The three figures are corresponding to $N_0=1\times10^5, 2\times10^5, 1\times10^6$ from left to right. }
\label{fig:simul_rmse_time}
\end{figure}

The \hl{denoising} results for one of the testing slices under $N_0=2\times10^5$ are given in figure \ref{fig:simul_images}. The parametric maps were almost unreadable if no filter was applied to the CTP images. Both Noise2Noise and supervised learning achieved significantly reduced noise in the \hl{time-frame}, CBF, CBV\hl{,} and MTT images, where the gain was the most significant in CBV and MTT images. In the CBF images, the small penumbra area pointed by the black arrows were severely distorted in the Gaussian, TIPS and TTV results, but were much better restored in the Noise2Noise and supervised learning results. Noise2Noise had similar CBF and CBV images compared to supervised learning, but it had slightly noisier MTT images due to the noise amplification of the dividing. 

Figure \ref{fig:simul_quantification} gives the RMSEs and SSIMs of CBF, CBV and MTT of the testing slices compared to noiseless results. Noise2Noise and supervised learning consistently outperformed the conventional methods for both RMSE and SSIM under all the three noise levels. Noise2Noise also had RMSE and SSIM close to supervised learning, although Noise2Noise did not have access to noiseless data during training. For $N_0=2\times10^5$ and $N_0=10^6$, supervised learning outperformed Noise2Noise by no more than 10\%, 15\% and 15\% in terms of RMSE of CBF, CBV and MTT, respectively. For $N_0=10^5$, supervised learning outperformed Noise2Noise by no more than 12\%, 15\% and 21\% for the RMSEs. 

Figure \ref{fig:simul_rmse_time} shows the change of RMSEs of denoised \hl{time-frame} images along the time. Unfiltered results were not included because they were significantly higher than the others. There was also no result for TTV since TTV does not denoise the \hl{time-frame} images. The peaks of the RMSE curves are corresponding to the peaks of AIF, where the image intensity greatly increased. Both Noise2Noise and supervised learning had significantly lower frame-wise RMSE compared to Gaussian filter and TIPS, whereas Noise2Noise had close RMSE to supervised learning. The difference between peak and baseline RMSEs was also greater for Noise2Noise and supervised learning compared to Gaussian and TIPS, which might be caused by the lack of training samples near \hl{the} peaks of time-concentration curves. 

\subsection{Influence of $\beta$}

\begin{figure}[t]
\centering
\includegraphics[width=\columnwidth]{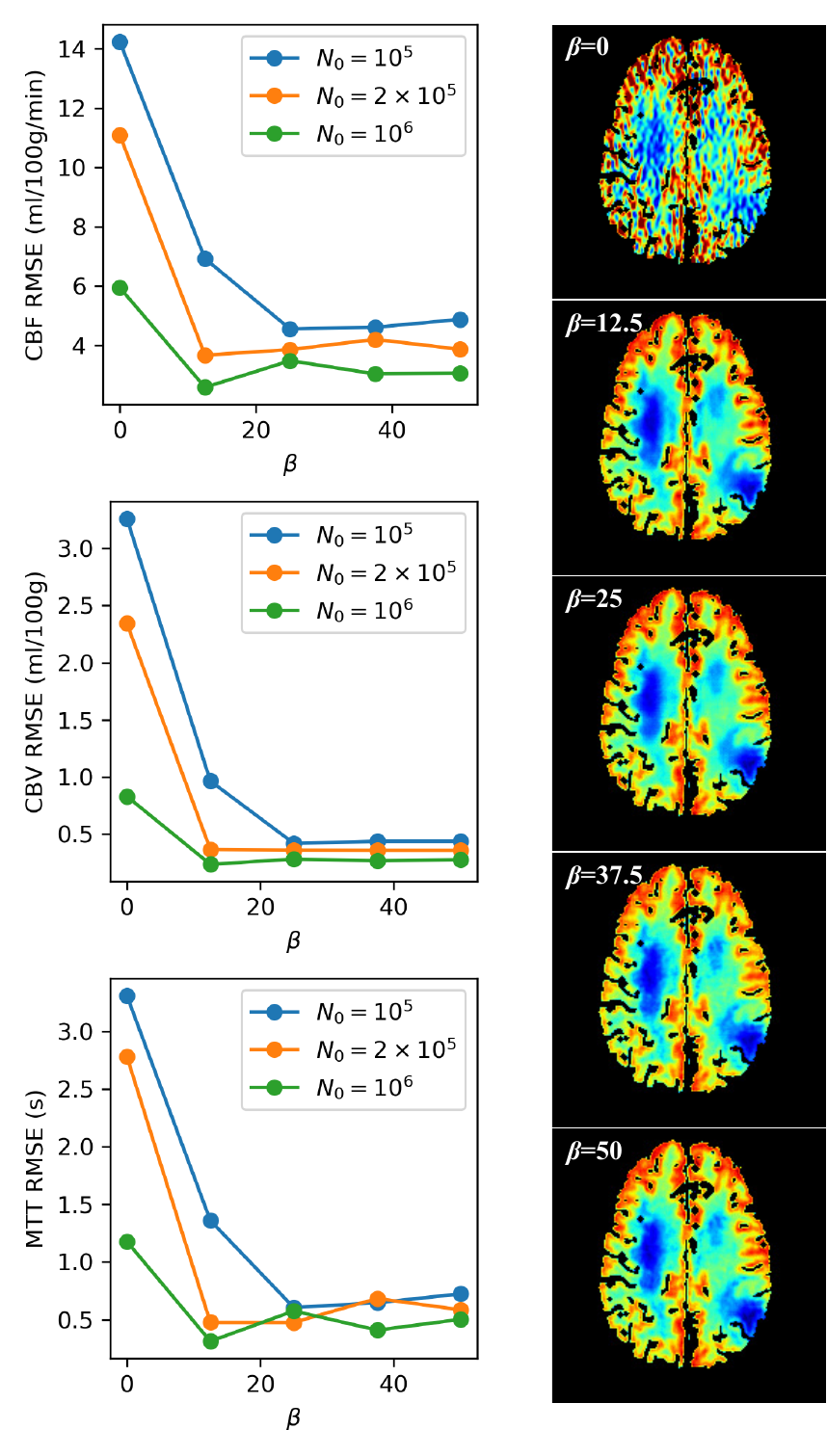}
\caption{Influence of $\beta$ for the denoising performance in simulation. The left column is the testing RMSEs of CBF, CBV and MTT under different $\beta$ and noise levels. The right column is the CBFs under $N_0=2\times10^5$ with different $\beta$ values.  }
\label{fig:simul_beta}
\end{figure}

The hyperparameter $\beta$ for the bias compensation term in Noise2Noise was tuned from 0 to 50 to investigate its influence on the denoising performance. The testing RMSEs of the simulation data under all three noise levels are given in figure \ref{fig:simul_beta}. The CBF images of one testing slice for $N_0=2\times10^5$ are also shown in the figure. 

The RMSEs generally decreased as $\beta$ increased. With large $\beta$ values, the RMSEs remained stable with a slightly increasing trend. There was no subtle visual difference of the CBF maps when $\beta$ was sufficiently large. The bias compensation term also acted as an efficient approach to avoid overfitting. As demonstrated by the CBF images in figure \ref{fig:simul_beta}, the network overfitted to noise when $\beta=0$ mainly due to lack of training samples. The overfitting was overcome with larger $\beta$ values. 

\subsection{Real Data Results}

\begin{figure*}[t]
\centering
\includegraphics[width = 0.9\linewidth]{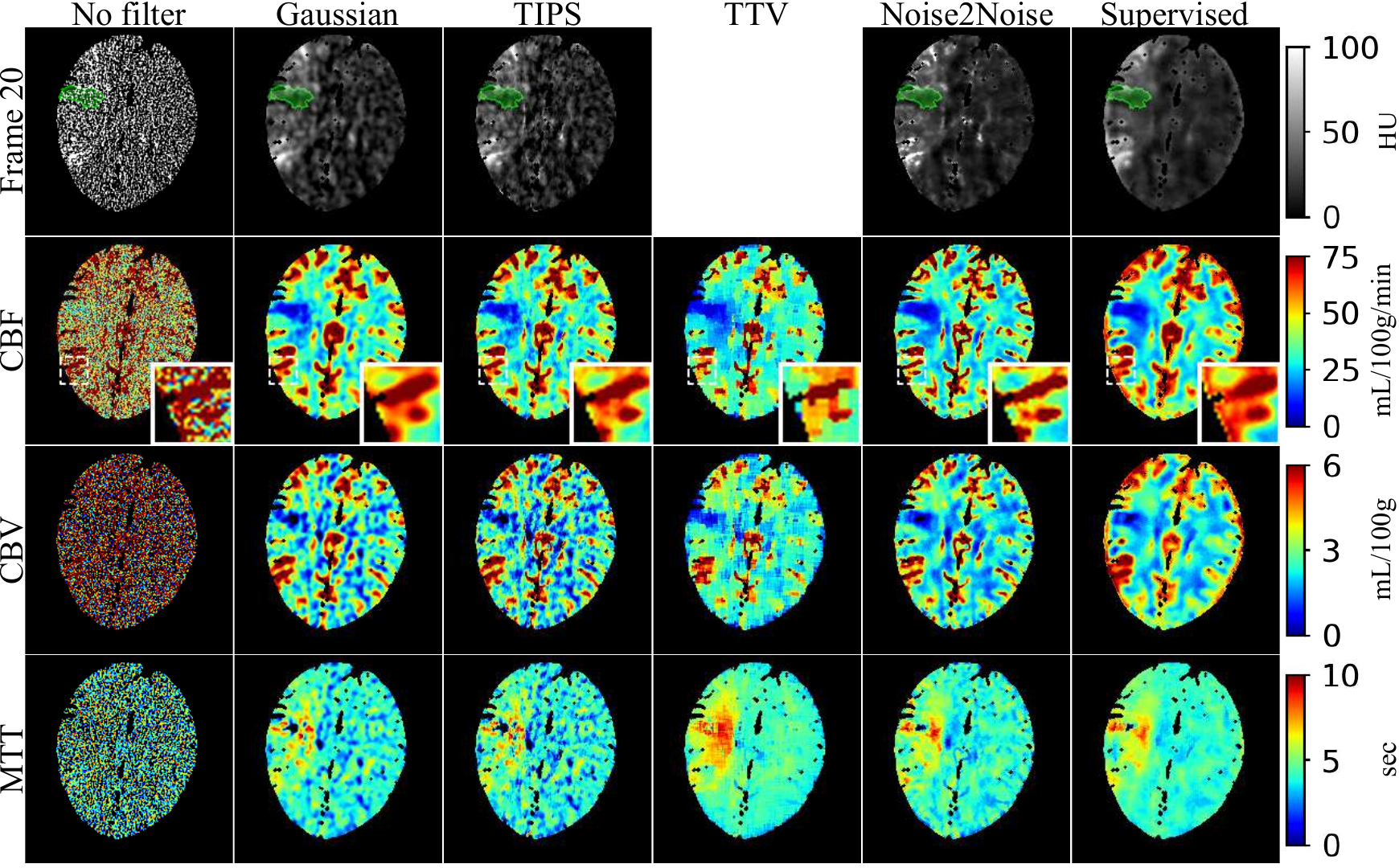}
\caption{Denoising results of a testing slice from the real dataset. The concentration maps at the 20th frame, CBF, CBV and MTT images are shown. A gray matter region is zoomed in for the CBF images to demonstrate the difference in spatial resolution. \hl{The time-concentration map} is not given for TTV because it does not denoise the time frames. \hl{The ischemic core derived from the DWI image is shown as the green overlay on the time-frame images. It should be noted that in supervised learning, the network was trained on simulation images and applied to real images.}}
\label{fig:real_images}
\end{figure*}

\begin{figure}[t]
\centering
\includegraphics[width=\columnwidth]{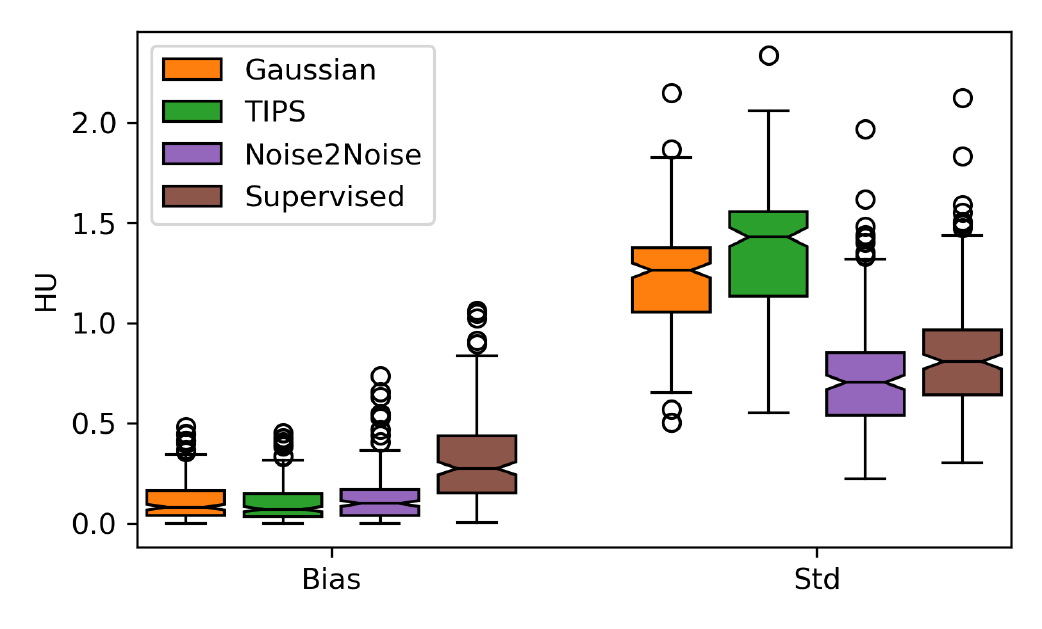}
\caption{The bias and std of the ROI inside normal white matter in all the denoised concentration maps. Each value at each time point was considered as a data point for the box plot. }
\label{fig:real_quantification}
\end{figure}

\begin{table}[t]
\renewcommand{\arraystretch}{1.3}
\caption{Mean CNRs of the testing CBF maps}
\label{table:cnr}
\centering
\begin{tabular}{llll}
\hline
Denoising Method	& Contrast 		& Std 			& CNR 	\\
				 	& (mL/100g/min)	& (mL/100g/min)	&		\\
\hline
No filter		 	& 48.27			& 150.07		& 0.30 	\\
Gaussian filter		& 100.76		& 32.32			& 2.83  \\
TIPS				& 100.04		& 34.20			& 2.64 	\\
TTV 				& 158.54		& 37.69			& 5.35 	\\
Noise2Noise 		& 100.25		& 25.90			& 3.29 	\\
Supervised learning	& 86.05			& 26.21			& 2.57 	\\
\hline
\end{tabular}
\end{table}

The denoising results of one testing slice are given in figure \ref{fig:real_images}. Similar to the simulation results in figure \ref{fig:simul_images}, unfiltered parametric maps are almost unreadable and denoising was necessary. CBF images from Gaussian, TIPS, TTV and Noise2Noise demonstrated similar image quality to the simulation results at $N_0=2\times10^5$ in figure \ref{fig:simul_images}. In the \hl{zoomed-in} gray matter areas of the CBF images, it could be observed that Noise2Noise had significantly improved spatial resolution compared to the other methods. Noise2Noise also had significantly reduced noise in CBV and MTT images compared to both Gaussian filter and TIPS. The TTV images had substantial blocky artifacts and structural bias of the gray matter compared to other methods. 

The supervised learning results had \hl{a} substantial blur of the gray matter compared to Noise2Noise, mainly due to the inconsistency between the simulation training data and the real testing data. The inconsistency also led to artifacts at the edge of cerebrum where the \hl{time-}concentration maps, CBF and CBV had larger value. 

Figure \ref{fig:real_quantification} \hl{shows} the bias and std of ROIs inside normal white matter in denoised concentration maps. Noise2Noise had slightly larger bias compared to Gaussian and TIPS, but significantly lower bias compared to supervised learning. Most frames had bias less than 0.5 HU for Noise2Noise. Noise2Noise had the least noise level among all the methods. There was no \hl{result} for TTV in figure \ref{fig:real_quantification} because it does not denoise the time frames. 

Table \ref{table:cnr} \hl{shows} the mean contrasts, stds, and CNRs of the testing CBF maps which were calculated according to section \ref{sec:metrics}. Noise2Noise had similar contrast with Gaussian filter and TIPS and the least std among all the methods. Unfiltered CBF had lower contrast compared to Gaussian and \hl{TIPS because the max operator} in the CBF calculation (\ref{eq:cbf}) led to non-zero-mean noise in the CBF images. Lower contrast was also observed for supervised learning results where bias was caused by inconsistent training and testing data. TTV had the best CNR because of the significantly larger contrast compared to the Gaussian filter and TIPS, which was mainly due to the overestimation of CBF of the reference white matter ROI. Noise2Noise had the best CNR among all the three methods (Gaussian, TIPS, Noise2Noise) without significant bias. 

\hl{
\subsection{Time Costs}
We further measured the testing time cost of all the methods and the training cost of Noise2Noise on the real dataset. The benchmarking was conducted on a computer with Intel Xeon Silver 4110 CPU @ 2.10GHz with 32 cores, a memory of approximately 97 GB, and a GPU of Nvidia GeForce GTX 1080 Ti.

The results are given in table \ref{table:time}. All the testing time costs were benchmarked on one CTP image with 1 slice and 44 frames. The "No filter" method counted only the time to calculate CBF, CBV, and MTT images on CPU. Gaussian, TIPS, and Noise2Noise (testing) included both time-frame denoising and parametric map calculation. TTV only included parametric map calculation since the denoising was embedded inside it. The Noise2Noise (training) is the network training time, where the number of training pairs was 4192, the batch size was 20, and the number of epochs was 100. 

It should be noted that the Noise2Noise had short testing time which could satisfy practical applications. The training time was also reasonable to reach good performance. 

\begin{table}[t]
\renewcommand{\arraystretch}{1.3}
\caption{\hl{Estimated time costs}}
\label{table:time}
\centering
\begin{tabular}{ll}
\hline
Method						& Time 			\\
\hline
No filter (CPU)		 		& 0.159s		\\
Gaussian filter (CPU)		& 0.233s		\\
TIPS (GPU)					& 1.098s	 	\\
TTV (CPU)					& 53.140s 	 	\\
Noise2Noise (Testing, GPU)	& 1.492s	 	\\
Noise2Noise (Training, GPU)	& $\approx$3 hours	 	\\
\hline
\end{tabular}
\end{table}
}

\section{Discussion and Conclusion}
\hl{In this paper,} we proposed a self-supervised learning method for dynamic CT perfusion image denoising based on the Noise2Noise principle. The main advantage of the method was that the training of the denoising deep neural network did not require \hl{high-quality} reference images, which are hard to acquire for CTP due to radiation dose concerns. It could overcome the problem of supervised learning that the performance will deteriorate when \hl{the testing data have a different distribution than the training data} since testing data itself can be used for training. Furthermore, it only required the time frame images which can be easily acquired, instead of projection data. 

The method achieved improved image quality on CBF, CBV and MTT images compared to denoising algorithms including Gaussian filter, TIPS and TTV on both simulation and real \hl{datasets}. \hl{In the simulation results, although supervised methods achieved better visual image quality and quantitative metrics, Noise2Noise was not significantly inferior. In the real data results, because clean training images are not available, supervised learning had dramatically reduced image quality by applying the network trained on the simulation data. Noise2Noise, on the other hand, maintained good performance and achieved improved spatial resolution and CNR over the supervised learning. }

\hl{It should be noted that one limitation of the Noise2Noise framework is the assumption of zero-mean noise. Non-zero-mean artifacts such as scatter, beam-hardening, metal artifacts, and motion artifacts can exist in the CTP images, which will breach the assumption for Noise2Noise. However, the bias in the noise will not lead to catastrophic failure in either training or testing. In principle, Noise2Noise training will converge to the average of the noise, so the bias part will be kept while noise is reduced. A brief analysis will be given in the appendix. Figure \ref{fig:real_images_bias} shows two testing slices with metal / motion artifacts. Despite the severe artifacts, there is no catastrophic failure of the network. It should be noted that the Noise2Noise results had biased estimation near the motion artifacts compared to TIPS. This is a network generalization problem that every deep learning method meets since motion did not appear in the training dataset. 

Another possible source of non-zero-mean noise comes from ultra-low-dose CT scans. The image reconstruction chain often leads to non-zero-mean noise in ultra-low-dose CT images, which may lead to biased parametric maps. In this study, we did not try further reduce the dose to below that in the 2018 ISLES challenges. However, current deconvolution-based parametric map estimation methods are already biased, due to the Tikhonov regularization during the deconvolution\cite{Fieselmann2011}. Scaling is commonly used to get the correct quantification\cite{Konstas2009, Wittsack2008, Leenders1990}. Furthermore, conventional denoising methods, including Gaussian filter and TIPS, cannot correct the bias in the noise. The proposed Noise2Noise will have a similar bias level compared to these methods. 

}

\begin{figure}[t]
\centering
\includegraphics[width=\columnwidth]{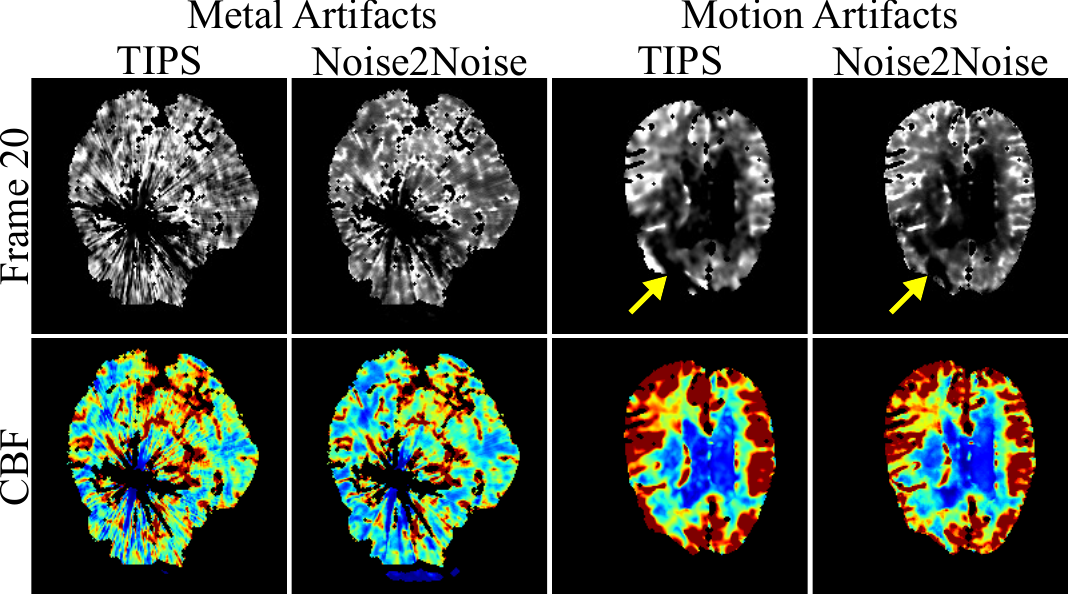}
\caption{\hl{Two slices of testing images with non-zero mean artifacts. The left two columns give a slice with metal artifacts. The metal implant in the middle of the brain was removed during preprocess. The right two columns give a slice with motion artifacts, which is marked by the yellow arrows on the time-frame images. The display windows are the same as that in figure \ref{fig:real_images}. }}
\label{fig:real_images_bias}
\end{figure}

\hl{The denoised images by Noise2Noise and supervised learning appears to be smoother than the ground truth. The main reasons are the high noise level in the original images and the L2-loss we used for the training. It is possible to further reduce the smoothing level by applying an additional penalty between the denoised images and the noisy images\cite{Wu2019} in the Noise2Noise. However, it may introduce additional noise and bias into the parametric maps. We did not add the extra penalty for maximal noise reduction and quantification accuracy. }

CTP has been demonstrated effective in selecting patients to receive mechanical thrombectomy within 6 to 24 hours from symptom onset based on DAWN and DIFFUSE 3 trials \cite{Powers2018, Nogueira2018, Albers2018}. However, CTP faces challenges of high radiation exposure and low spatial resolution compared to DWI \hl{despite} its advantage in CT's availability \cite{Gonzalez2012}. CTP has also been demonstrated to have low sensitivity to lacunar strokes which have small size and composed approximately 20\% of all the ischemic stroke, partially due to the artifacts in the parametric maps \cite{Biesbroek2013, Benson2016}. A recent study demonstrated the effectiveness of intravenous thrombolysis on lacunar stroke within 4.5 hours of symptom onset, which raised the potential needs for fast detection of lacunar stroke \cite{Barow2019}. There is \hl{the possibility} that the improved image quality and spatial resolution brought by the proposed deep learning approach could improve the utility of CTP in stroke treatment. 

The proposed method can be further improved to include more frames in the input to further reduce the noise level of denoised time frame images by using a recurrent neural network. It is also possible to expand the framework to other dynamic imaging scenarios such as myocardial CT perfusion, arterial spin label imaging, or dynamic positron emission tomography \cite{Ho2010, Alsop1998, Gong2018}.

\hl{
\section{Appendix}
\subsection{Details on Calculation of Parametric Maps}\label{app:ctp_method}
We used \hl{the deconovlution method} to calculate the parametric maps \cite{Fieselmann2011}. \hl{Given $\textrm{AIF}(t)$}, which is the time-concentration curve inside the arteries that can be automatically estimated from CTP images \cite{Mouridsen2006, Kao2014}, the time-concentration curve of voxel $j$, $c_j(t)$, can be expressed as a convolution:
\begin{equation}\label{eq:c_aif}
c_j(t) = (\mathrm{AIF} * r_j)(t), 
\end{equation}
where $r_j(t)$ is the response function of the tissue to pulse input. The parametric maps can be calculated from $r_j(t)$ as:
\begin{align}
\mathrm{CBF}_j &= \frac{1}{\rho_j}\max_t(r_j(t)) \label{eq:cbf}\\
\mathrm{CBV}_j &= \frac{1}{\rho_j}\sum_t r_j(t)\Delta t \\
\mathrm{MTT}_j &= \frac{\mathrm{CBV}_j}{\mathrm{CBF}_j} \\
\mathrm{TTP}_j &= \arg\max_t c_j(t)\Delta t 
\end{align}

$r_j(t)$ can be calculated by solving the matrix form of (\ref{eq:c_aif}). It was demonstrated in \cite{Wu2003, Wittsack2008} that it could correct the local delay of AIF by building a block-circulant convolution matrix, $\mathbf{A} \in \mathcal{R}^{M\times M}, M \geq 2T$. The length of $\mathbf{c}(t)$ and $\mathrm{AIF}(t)$ was first increased to $M$ by zero-padding, and matrix $\mathbf{A}$ was defined as:
\begin{equation}
A_{ij} = \Big\{
\begin{array}{ll}
\mathrm{AIF}(i-j+1)\Delta t, & i\leq j \\
\mathrm{AIF}(M+i-j+1)\Delta t, & i > j
\end{array}
\end{equation} 

The deconvolution problem (\ref{eq:c_aif}) became solving the following linear equation:
\begin{equation}\label{eq:convolution}
\mathbf{c}_j = \mathbf{Ar}_j,
\end{equation}
where $\mathbf{c}_j = (c_j(1), \dots, c_j(T), 0, \dots, 0)^T \in \mathcal{R}^M$. 

Equation (\ref{eq:convolution}) was solved via singular value decomposition (SVD) with Tikhonov regularization \cite{Fieselmann2011}, which actually solved the following problem:
\begin{equation}
\hl{\mathbf{r}_j = \arg\min_{\mathbf{r}_j'} \norm{\mathbf{Ar}_j' - \mathbf{c}_j}_2^2 + \lambda^2 \norm{\mathbf{r}_j'}_2^2,}
\end{equation}
where $\lambda=\lambda_\textrm{rel}\sigma_\textrm{max}$, and $\sigma_\textrm{max}$ is the largest singular value of $\mathbf{A}$. We used $\lambda_\textrm{rel}=0.3$ according to \cite{Fieselmann2011}. An additional scaling factor were applied to CBF and CBV to correct bias. In the simulation, the scaling factors minimized the L2 distance between calculated CBF and CBV maps with the ground truth. In the real data, they were chosen so that a region-of-interest (ROI) selected in the normal white matter had average CBF and CBV of 22 mL/100g/min and 2 mL/100g, respectively \cite{Konstas2009, Wittsack2008, Leenders1990}.

\subsection{Proof of Noise2Noise}\label{app:noise2noise}
We will provide a proof that the training cost function of Noise2Noise is equivalent to training with clean images. The proof is basically the same with \cite{Wu2019}, but we put it here for the self-consistency of the paper. 

\begin{thm}\label{thm:noise2noise}
The following equation holds:
\begin{equation}\label{eq:thm1}
\begin{split}
	& \frac{1}{N}\sum_i\norm{f(\mathbf{x}_i + \mathbf{n}_{i1};\mathbf{\Theta}) - (\mathbf{x}_i + \mathbf{n}_{i2})}_2^2 \\
=  & \frac{1}{N}\sum_i\norm{f(\mathbf{x}_i + \mathbf{n}_{i1};\mathbf{\Theta}) - \mathbf{x}_i}_2^2 + C,
\end{split}
\end{equation}
where $C$ is irrelevant to $\mathbf{\Theta}$, if the following conditions are satisfied:

1. $N\rightarrow\infty$;

2. Conditional expectation $E\{\mathbf{n}_{i2} | \mathbf{x}_i\} = 0$;

3. $\mathbf{n}_{i1}$ and $\mathbf{n}_{i2}$ are independent;

4. $\forall i$, $f(\mathbf{x}_{i} + \mathbf{n}_{i1}; \mathbf{\Theta}) < \infty$, $\mathbf{n}_{i2}<\infty$

\end{thm}

\begin{proof}
For simplicity, let $\mathbf{f}_i = f(\mathbf{x}_i + \mathbf{n}_{i1}; \mathbf{\Theta})$, then we have the following for the left hand side of (\ref{eq:thm1}):
\begin{equation}\label{eq:expand}
\begin{split}
& \frac{1}{N}\sum_i\norm{\mathbf{f}_i - (\mathbf{x}_i + \mathbf{n}_{i2})}_2^2 \\
= &\frac{1}{N}\sum_i\norm{\mathbf{f}_i - \mathbf{x}_i}_2^2 - \frac{1}{N}\sum_i 2\mathbf{n}_{i2}^T\mathbf{f}_i + \\
 &\frac{1}{N}\sum_i (\mathbf{n}_{i2}^T\mathbf{n}_{i2} + 2\mathbf{n}_{i2}^T\mathbf{x}_i).
\end{split}
\end{equation}

The last term is the $C$ in (\ref{eq:thm1}) because it is irrelevant to $\mathbf{\Theta}$. Then the only difference between (\ref{eq:expand}) and the right hand side of (\ref{eq:thm1}) is the second term. Because $N\rightarrow\infty$, according to Lindeberg-Levy central limit theorem:
\begin{equation}
\frac{1}{N}\sum_i 2\mathbf{n}_{i2}^T\mathbf{f}_i \xrightarrow{d} \mathcal{N}(E\{2\mathbf{n}_{i2}^T\mathbf{f}_i\}, \frac{1}{N}\sigma^2\{2\mathbf{n}_{i2}^T\mathbf{f}_i\}),
\end{equation}
where $E\{\cdot\}$ is the expectation, $\sigma^2\{\cdot\}$ is the variance, and $\mathcal{N}(\mu, \sigma^2)$ is a Gaussian distribution with mean $\mu$ and variance $\sigma^2$. 

Because both $\mathbf{f}_i$ and $\mathbf{n}_{i2}$ are bounded by condition 4, $\sigma^2\{2\mathbf{n}_{i2}^T\mathbf{f}_i\}$ is bounded so $\sigma^2\{2\mathbf{n}_{i2}^T\mathbf{f}_i\} / N \rightarrow 0$. Hence, the Gaussian distribution will converge to its mean value, which means:
\begin{equation}\label{eq:limit_of_second_term}
\frac{1}{N}\sum_i 2\mathbf{n}_{i2}^T\mathbf{f}_i \rightarrow E\{2\mathbf{n}_{i2}^T\mathbf{f}_i\} = 2E\{\mathbf{f}_{i}^T E\{\mathbf{n}_{i2} | \mathbf{f}_{i}\}\}
\end{equation}

Because $\mathbf{f}_i$ is a deterministic function of $\mathbf{x}_i$ and $\mathbf{n}_{i1}$, we have:
\begin{equation}\label{eq:conditional_expection}
E\{\mathbf{n}_{i2}|\mathbf{f}_i\} = E\{\mathbf{n}_{i2}|\mathbf{x}_i, \mathbf{n}_{i1}\}
\end{equation}

Because $\mathbf{n}_{i2}$ is independent from $\mathbf{n}_{i1}$, we have:
\begin{equation}\label{eq:expection_0}
E\{\mathbf{n}_{i2}|\mathbf{x}_i, \mathbf{n}_{i1}\} = E\{\mathbf{n}_{i2}|\mathbf{x}_i\} = 0
\end{equation}

Substitute (\ref{eq:expection_0}) into (\ref{eq:conditional_expection}) and (\ref{eq:limit_of_second_term}), we have:
\begin{equation}
\frac{1}{N}\sum_i 2\mathbf{n}_{i2}^T\mathbf{f}_i \rightarrow 0, 
\end{equation}
which infers that the second term on the right hand side of (\ref{eq:expand}) is zero. This concludes the proof.

\end{proof}

Note that in theorem \ref{thm:noise2noise}, condition 1 is the assumption for most learning-based algorithms; condition 4 can be easily satisfied by common networks (including UNet) and realistic noise; condition 2 is the zero-mean property of the noise; condition 3 is the requirement for the independence between the two noise realizations, which is the main focus of this work. 

To analyze how noise correlation and bias will influence Noise2Noise training, we can directly substitute (\ref{eq:conditional_expection}) and (\ref{eq:limit_of_second_term}) to (\ref{eq:expand}) and reach:
\begin{equation}\label{eq:bias_analysis}
\begin{split}
& \frac{1}{N}\sum_i\norm{\mathbf{f}_i - (\mathbf{x}_i + \mathbf{n}_{i2})}_2^2 \\
= &\frac{1}{N}\sum_i\norm{\mathbf{f}_i - \mathbf{x}_i}_2^2 - \frac{1}{N}\sum_i 2\mathbf{f}_i^TE\{\mathbf{n}_{i2} | \mathbf{x}_i, \mathbf{n}_{i1}\} + C \\
= &\frac{1}{N}\sum_i\norm{\mathbf{f}_i - (\mathbf{x}_i + E\{\mathbf{n}_{i2} | \mathbf{x}_i, \mathbf{n}_{i1}\})}_2^2 + C_1,
\end{split}
\end{equation}
where $C$ and $C_1$ are irrelevant to $\mathbf{\Theta}$. 

If $\mathbf{n}_{i2}$ is biased but independent from $\mathbf{n}_{i1}$ (breach of condition 2), then $E\{\mathbf{n}_{i2} | \mathbf{x}_i, \mathbf{n}_{i1}\} = g(\mathbf{x}_i)$ where $g(\mathbf{x}_i)$ is a deterministic function. Then instead of converging to $\mathbf{x}_i$, Noise2Noise will converge to $\mathbf{x}_i + g(\mathbf{x}_i)$. $\mathbf{x}_i + g(\mathbf{x}_i)$ is equivalent to taking multiple measurement and averaging, which is the most straight-forward way of noise reduction and works well under most conditions. Hence, we did not consider the breach of condition 2 as the biggest challenge to the efficacy of Noise2Noise. 

If $\mathbf{n}_{i2}$ is zero-mean but correlated with $\mathbf{n}_{i1}$ (breach of condition 3), we are yet to find a general formula to describe the consequence. In the simplest case, where $\mathbf{n}_{i1}$ and $\mathbf{n}_{i2}$ are white Gaussian noise with correlation $c$, we have $E\{\mathbf{n}_{i2} | \mathbf{x}_i, \mathbf{n}_{i1}\} = c\mathbf{n}_{i1}$. Then the Noise2Noise training will converge to $\mathbf{x}_{i} + c\mathbf{n}_{i1}$, which poses a big challenge because the noise is not completely removed. A more comprehensive analysis is beyond the scope of this work. 
}



\ifCLASSOPTIONcaptionsoff
  \newpage
\fi



\bibliographystyle{IEEEtran}

\bibliography{ms}
\end{document}